\documentclass[a4paper,oneside,11pt]{article}

\usepackage{calrsfs}
\usepackage{amsmath,amsfonts,amscd,amssymb,amsthm}
\usepackage[english]{babel}
\usepackage[utf8]{inputenc}
 \usepackage[active]{srcltx}
 \usepackage[T1]{fontenc}
\usepackage{graphicx}
\usepackage{enumitem}
\usepackage{bbm}
\usepackage{calrsfs}
\usepackage{mathtools}
\usepackage{color}
\usepackage{MnSymbol}
\usepackage{stmaryrd}

\newtheorem{theorem}{Theorem}[section]

\newtheorem{corollary}[theorem]{Corollary}
\newtheorem{lemma}[theorem]{Lemma}
\newtheorem{proposition}[theorem]{Proposition}
\newtheorem{definition}[theorem]{Definition}

\newtheorem{remark}[theorem]{Remark}

\newcommand{\be}[1]{\begin{equation}\label{#1}}
\newcommand{\ee}{\end{equation}}
\numberwithin{equation}{section}

\newcommand{\ba}[1]{\begin{align}\label{#1}}
\newcommand{\ea}{\end{align}}
\numberwithin{equation}{section}

\newcommand{\ben}{\begin{equation*}}
\newcommand{\een}{\end{equation*}}
\numberwithin{equation}{section}

\newcommand{\calA}{\mathcal{A}}

\newcommand{\calC}{\mathcal{C}}

\newcommand{\calE}{\mathcal{E}}
\newcommand{\calF}{\mathcal{F}}

\newcommand{\calK}{\mathcal{K}}

\newcommand{\calM}{\mathcal{M}}
\newcommand{\calN}{\mathcal{N}}

\newcommand{\calP}{\mathcal{P}}

\newcommand{\lr}[1]{\xleftrightarrow{#1}}

\newcommand{\bbC}{\mathbb{C}}

\newcommand{\bbG}{\mathbb{G}}

\newcommand{\bbN}{\mathbb{N}}

\newcommand{\bbP}{\mathbb{P}}

\newcommand{\bbR}{\mathbb{R}}

\newcommand{\bbT}{\mathbb{T}}

\newcommand{\bbZ}{\mathbb{Z}}

\newcommand{\n}{\mathbf{n}}
\newcommand{\m}{\mathbf{m}}

\def\indf#1{\mathbf 1\left[#1\right]}   

\renewcommand{\be}{\begin{equation} }  
\renewcommand{\ee}{\end{equation} } 

\usepackage{color}
\def\Z{{\mathbb Z}}

\title{Random Currents and Continuity of Ising Model's Spontaneous Magnetization} 
\author{Michael Aizenman, Hugo Duminil-Copin,  Vladas Sidoravicius  \\ 
 }
\date{
\hspace{8.5cm}  {\em To Ruty (1987 - 2014)} \qquad 
\\[5ex]
(Ver. 20 Feb 2014; notation adj.  29 June 2014)}

\begin{document}
\maketitle

\begin{abstract}  
The spontaneous magnetization is proved to vanish continuously  at the critical temperature for  a class of ferromagnetic Ising spin systems  which includes the nearest neighbor ferromagnetic Ising spin model on $\bbZ^d$  in $d=3$ dimensions.    The analysis applies also to higher dimensions, for which the result is already known, and to systems with interactions of power law decay.    
The proof employs in an essential way an extension of Ising model's {\it random current representation} to the model's infinite  volume limit.  Using it, we relate  the  continuity of the magnetization to the vanishing of the free boundary condition Gibbs state's Long Range Order parameter.  For reflection positive models the resulting criterion for continuity may be established through the  infrared bound for all but the borderline lower dimensional cases.  The exclusion applies to the one dimensional model with $1/r^2$ interaction for  which the spontaneous magnetization is known to be discontinuous at $T_c$. 
 
\end{abstract}


\section{Introduction}

The Ising model needs no introduction, being perhaps the most studied example, since its formulation by Lenz~\cite{Len20}, of a system undergoing a phase transition.    The transition it exhibits was found to be of rather broad relevance, though of course its features do not exhaust the range of possible behaviors in statistical mechanics.   The model has provided the testing ground for a large variety of techniques, which partially compensate for the lack of exactly solvable models above two dimensions.   

Our goal here is to present a tool for addressing the question of continuity of the model's phase transition.     Doing so, we advance  the technique of  the model's 
\emph{random current} representation which was developed in \cite{Aiz82} starting from  the Griffiths-Hurst-Sherman switching lemma  (which was earlier used in~\cite{GHS70} for the GHS inequality).  In this representation the onset of the Ising model's  symmetry breaking is presented as a percolation transition in a system of random currents with constrained sources.   The perspective that this picture  offers has already shown itself to be of value in yielding a range of results for the model's critical behavior (c.f.~\cite{Aiz82,AF86,ABF87,Sak07}).  The incremental step taken here is to consider directly the limiting shift invariant infinite systems of random currents.  This allows to add to the available tools arguments based on  the `uniqueness of infinite cluster' principle, which is of relevance to the question of  continuity of the state at the model's critical temperature. 

Among the specific cases for which the results presented here answer a long open question is the continuity at the critical point of the spontaneous magnetization of the standard three dimensional Ising model.   To highlight this it may be noted that what is proven here for Ising spin systems is not valid for  the broader class of the $Q$-state Potts models, with counter-examples existing  for $Q>4$ (Ising case corresponding to $Q=2$).    

\subsection{Notation} 
We focus here on the $d$-dimensional version of the ferromagnetic Ising model, on the transitive graph $\Z^d$.  More generally, the model may be formulated  on  a graph,  whose vertex set and edge set we denote  by $\bbG$ and $E\subset \bbG^2$ correspondingly.  Associated with the sites are $\pm1$ valued spin variables, whose configuration is denoted $\sigma=(\sigma_x:x\in \bbG)$.  

For a general ferromagnetic pair interaction, the system's Hamiltonian defined for finite subsets $\Lambda \subset \bbG$ and boundary conditions $\tau \in \{-1,0,1\}^ {\bbG \backslash \Lambda}$ is given by the function 
\be   \label{eq:H}
H_\Lambda^\tau(\sigma)~:=~  -\sum_{x\in \Lambda} h \sigma_x \ -\ \sum_{\{x,y\}\subset \Lambda: x\neq y}J_{x,y}\sigma_x\sigma_y  \ -\  \sum_{x\in  \Lambda: y \in \bbG \backslash \Lambda}J_{x,y}\sigma_x\tau_y   \,,
\ee 
for any $\sigma\in\{-1,1\}^{\Lambda}$, where $(J_{x,y})_{x,y\in\bbZ^d}$ is a family of nonnegative {\em coupling constants}, and $h$ is the magnetic field.

 For $\beta \in (0, \infty)$, finite volume Gibbs states with boundary conditions $\tau$ are given by  probability measures on the spaces of configurations in finite subsets $\Lambda \subset \bbG$ under which the expected values of  functions $f:\{-1,1\}^{\Lambda}\rightarrow \bbR$ are
$$\langle f\rangle_{\Lambda,\beta,h}^\tau=\sum_{\sigma\in\{-1,1\}^{\Lambda}}f(\sigma)\frac{e^{-\beta H^\tau_\Lambda(\sigma)}}{Z^\tau(\Lambda,\beta,h)}\,  ,$$
where the sum is normalized by the partition function $Z^\tau(\Lambda,\beta,h)$   so that $\langle 1\rangle_{\Lambda,\beta,h}^\tau=1$. Of particular interest is the following pair of boundary conditions (b.c.): 
\begin{itemize} [noitemsep,nolistsep]
\item {\em free b.c.}: $\tau_x = 0 $ for all $x\in \bbG\backslash \Lambda$ (or alternatively, the last term in \eqref{eq:H} is omitted).
\item  {\em plus   b.c.}: $\tau_x =  1$  for all $x\in \bbG \backslash \Lambda$.  
\end{itemize} 
The corresponding measures, or expectation value functionals,  are denoted  $  \langle  \cdots    \rangle_{\Lambda,\beta,h}^0$ and   
$  \langle  \cdots    \rangle_{\Lambda,\beta,h}^+$ (with ``$ \cdots $'' a place holder for functions of the spin configurations).
 
For each of these two boundary conditions, the finite volume Gibbs states are known to  converge  to the corresponding infinite-volume Gibbs measures  (for a discussion of the concept see e.g. \cite{Georgii}). 
By default the volume subscript will be omitted when it refers to the full graph, i.e. $\Lambda = \bbG$.

For simplicity we   focus here on the prototypical example  of $\bbG= \bbZ^d$, and  interactions which are: 
\vspace{-.4cm} 
\begin{enumerate} [noitemsep]
\item[\bf C1] translation invariant:  $J_{x,y}=J_{0,y-x}$,
\item[\bf C2] ferromagnetic:   $J_{x,y}\ge 0$,
\item [\bf C3] locally finite: $|J|:=\sum_{x\in\bbZ^d}J_{0,x}<\infty$.
\item [\bf C4] aperiodic:  for any $x\in \bbZ^d$, there exist $0=x_0,x_1,\dots,x_{m-1},x_m=x$ such that $J_{x_0,x_1}J_{x_1,x_2},\dots , J_{x_{m-1},x_m}>0$.
\end{enumerate}
 The last condition is benign since under {\bf C1}   the lattice can be divided into sub-lattices with {\bf C4}   holding for each sub-lattice, and the arguments presented below can be  adapted to such setup   
\footnote{Among the essential features of the graphs $\bbZ^d$ is their transitivity and sub-exponential growth.   Our arguments can be extended, however not to graphs of positive Cheeger constant, such as regular trees and more generally nonamenable Cayley graphs.   Nor are the statements proven below valid at such generality (the percolation aspect of this distinction  is discussed in \cite{SNP00} and references therein).}.   \\ 

Of particular interest is the model's phase transition, which in the $(\beta,h)$ plane occurs along the $h=0$ line and is reflected in the nonvanishing of the symmetry breaking order parameter: 
\be 
m^*(\beta) \ := \ \langle\sigma_0\rangle_{\beta}^+  \, .    
\ee 
For temperatures ($T\equiv \beta^{-1}$) at which $m^*(\beta) >0$, the mean magnetization at nonzero magnetic field $h$ changes discontinuously at $h=0$.  The discontinuity  is symptomatic of the co-existence of  two distinct Gibbs equilibrium states :
\begin{eqnarray} 
  \langle  \cdots    \rangle_{\beta}^+ &= &  \lim_{h\searrow 0} \langle  \cdots    \rangle_{\beta,h}  \notag \\ 
\mbox{} \\ 
\langle  \cdots    \rangle_{\beta}^- &= &  \lim_{h\nearrow 0} \langle  \cdots    \rangle_{\beta,h}  \notag 
\end{eqnarray}  
which carry the residual magnetizations:
\be 
\langle \sigma_0  \rangle_{\beta}^\pm \ = \  \pm m^*(\beta) \, ,   
\ee 
with  $m^*(\beta)$ customarily referred to as the 
 {\em spontaneous magnetization}.   \\

 Property $\bf C3$ guarantees that at small $\beta$ (in particular, for $\beta< |J|^{-1}$, see e.g.~\cite{Fisher,Dob70,Aiz82})  
 $m^*(\beta)=0$, and  there is no symmetry breaking.    However, in dimensions $d>1$ each such model exhibits a phase transition at some 
$\beta_c  \in (|J|^{-1} , \infty)$, with $ m^*(\beta)>0$ for $\beta> \beta_c$~\cite{Pei36}.  For $d=1$ such a transition occurs  if $J_{x-y} \ge 1/|x-y|^\alpha$ with  $\alpha \in (1,2)$ \cite{Dys69} and also for the boundary value $\alpha = 2$~\cite{ACCN88}, in which case $m^*(\beta)$ is discontinuous at $\beta_c$~\cite{Tho69,ACCN88}).  \\

\subsection{The continuity question} 

The main result presented here addresses the following two related questions concerning the continuity of the correlation functions at this transition: 
\begin{enumerate} 
\item[Q1:]  Is $m^*(\beta)$ continuous at $\beta_c$, or equivalently  
is\footnote{Being the limit of a decreasing sequence of (finite volume) continuous functions, $m^*(\beta)$ is upper semicontinuous, and hence  
$m^*(\beta)  =  \lim_{\varepsilon \searrow 0} \  m^*(\beta + \varepsilon)$ for all $\beta \ge 0$ \cite{LML72}.}:    
\be  \label{eq:cont_m}
  \lim_{\beta \searrow \beta_c} \   m^*(\beta) \ = \  0 \, ?
\ee 
\item[Q2:]  Is the   Gibbs state $  \langle  \cdots    \rangle^+_\beta$ continuous at $\beta_c$?  
 \end{enumerate} 

By the arguments of \cite{Leb77, Leb72}  
for ferromagnetic Ising models in the class discussed here the answers to the two questions is the same.        (For completeness the argument is summarized here in Appendix~\ref{app:continuity}.)

As it is often the case, questions which at the level of Statistical Mechanics concern continuity of Gibbs states  have  Thermodynamic level manifestation in terms of differentiability properties of the   \emph{pressure} 
\be  \label{eq:pressure}
P(\beta, h) \ := \  \lim_{L \to \infty} \frac{1} {|\Lambda_L|} \log Z^\tau(\Lambda_L,\beta,h),
\ee
where $\Lambda_L=[-L,L]^d$ (with $\frac{-1}{\beta} P(\beta, h)$ also referred to as the \emph{free energy}). 
 On general grounds, the limit \eqref{eq:pressure} is known to: i) exist, ii) be independent of the boundary conditions, and iii) yield a jointly convex function of $(\beta, \beta h) $. 
   
 For ferromagnetic Ising models it is  known that  for $h\neq0$ the function $P(\beta, h) $ is analytic in $\beta$ and $h$~\cite{LeeYang} and thus the continuity questions are limited to the line $h=0$.

 J.L. Lebowitz~\cite{Leb77} proved that for any $\beta$ the differentiability of the free energy with respect to $\beta$ (in the class of models described above) is equivalent to the continuity of the expectation values of even spin functions (i.e. functions which are invariant under global spin flip) averaged over any of the translation invariant Gibbs states, or also equivalently  to the relation 
\be  \label{eq:f=pm}
   \langle  \cdots    \rangle^0_{\beta_c} \ = \ \frac{1}{2} \left[  \langle  \cdots    \rangle^+_{\beta_c} +  \langle  \cdots    \rangle^-_{\beta_c} \right],
\ee 
where $  \langle  \cdots    \rangle^-_{\beta_c}$ is the Gibbs measure with  {\em minus   b.c.}, or simply the image of 
  $  \langle  \cdots    \rangle^+_{\beta_c}$ under global spin flip.   
Furthermore, it was pointed out there that by convexity  the  differentiability condition can fail at only a countable set of $\beta$.    For finite range models T.~Bodineau~\cite{Bod06} reduced that to the one point set consisting of  just  the critical point $\beta_c$, and furthermore showed that the magnetization is also continuous at all $\beta\neq \beta_c$. (Some further related results are mentioned below.) 
It is generally expected that for all Ising models of the kind described above ([{\bf C1 - C4}])  equation  \eqref{eq:f=pm}  is  met also at $\beta_c$.  

When condition  \eqref{eq:f=pm} holds that it is not due to most general thermodynamic reason, as it fails for the ferromagnetic nearest-neighbor Potts models with $Q$ large enough: $Q>4$ in $d=2$ dimensions (see \cite{Bax73} for exact results, and \cite{KS,LMR,DST13,DC_pf} for partial rigorous results) and $Q>2$ for $d\ge3$ (see \cite{KS,BCC} for partial mathematical results in this direction).   
The failure of  \eqref{eq:f=pm} implies that \eqref{eq:cont_m} also fails, i.e $m^*(\beta_c) \neq 0$ (as is explained in Appendix~\ref{app:continuity}).  However the converse does not seem to be tautologically true, as is indicated\footnote{It is expected, but not proven, that in this  borderline case \eqref{eq:f=pm} is satisfied, while $m^*(\beta_c)\neq 0$.}
 by the special case of the one-dimensional  model with $J_{x,y} = 1/|x-y|^2$, c.f. \cite{Tho69,ACCN88}).

\subsection{Statement of the main results} 

Relevant to  the continuity of the spontaneous magnetization is the Long Range Order (LRO) parameter  which is defined by:    
\be 
  M_{LRO}(\beta)^2 \ := \  \lim_{n\to \infty}  \frac{1}{|\Lambda_n|} \sum_{x\in \Lambda_n} \langle\sigma_0\sigma_x\rangle_{\beta}^0   \, ,  
\ee 
 where $\Lambda_n=[-n,n]^d$, for the model on $\Z^d$ (the limit existing by monotonicity arguments),  or the LRO parameter's  variant 
\be 
\widetilde M_{LRO}(\beta)^2  \  :=  \  \inf_{B\subset \Z^d, |B|<\infty}    \frac{1}{|B|^2} \sum_{x,y \in B}  \langle  \sigma_x \sigma_y   \rangle_{\beta}^0   \ \equiv  \  \inf_{B\subset \Z^d, |B|<\infty}  \left \langle \left[\frac{1}{|B|} \sum_{x\in B} \sigma_x \right]^2 \right \rangle_{\beta}^0 \,   
\ee 
which satisfies 
\be 
\inf_{x\in \bbZ^d} \langle \sigma_0 \sigma_x \rangle^0_\beta  \ \le \  \widetilde M_{LRO}(\beta)^2 \ \le \   M_{LRO}(\beta)^2 \,.
\ee    

It may be noted that whereas 
 $m^*(\beta_c)$ provides  direct information about the states at $\beta > \beta_c$, the monotonicity arguments of~\cite{Leb77}   imply that 
$M_{LRO}(\beta_c)$ provides  direct information about the states at $\beta < \beta_c$ and, furthermore,   the following relation holds.

\begin{proposition}  For any translation invariant ferromagnetic Ising model on $\Z^d$:
at all $\beta \ge 0$    
\be    M_{LRO}(\beta) \  \le \   m^*(\beta)  
\ee 
with equality  holding at values of $\beta$ at which $P(\beta,0)$ is continuously differentiable. 
\end{proposition} 

Our main general result is: 
   
\begin{theorem}\label{thm:continuity}
For any ferromagnetic Ising model on $\bbZ^d$ whose coupling constants $(J_{x,y})_{x,y\in\bbZ^d}$  satisfy the conditions {\bf C1-C4}: if  
\be  \label{eq:cont_cond}
 \widetilde M_{LRO}(\beta) \ = \ 0 \ 
\ee 
then also 
\be 
 m^*(\beta_c) \ =\   0 \, , 
 \ee  
 and the system has only one Gibbs state at $\beta_c$.   
 \end{theorem}
\noindent (In which case one may add that the  Gibbs states $\langle \cdot  \rangle^{\#}$ are  continuous in  $\beta$ at  $\beta_c$, regardless of the choice of the boundary conditions, see Appendix~\ref{app:continuity}.)

\noindent{\bf Remarks:} 
\begin{enumerate}   
\item[i.] The reference to  $ \widetilde M_{LRO}(\beta)$ in \eqref{eq:cont_cond} instead of $ M_{LRO}(\beta)$ allows to 
base the  estimate on the spin-spin correlations along just one of the principal axes.  Kaufman and Onsager found such  simplification  helpful  in  calculations, by  which they showed that $\widetilde M_{LRO}(\beta_c)=0$  for the  {\it nearest neighbor} Ising model in two dimensions~\cite{KO49}.  

\item[ii.]   
Theorem~\ref{thm:continuity} carries  interesting implication   
also for the case that   the magnetization is discontinuous at $\beta_c$ (as in the above mentioned $1D$ model with the borderline $1/r^2$  interaction).  
It shows that even in such case the Ising model  does not display  coexistence at $\beta_c$ of two distinct phases: a `high temperature phase' with rapid decay of correlations and a `low temperature phase' exhibiting spontaneous magnetization, $m^*(\beta_c)>0$.   Such phase coexistence would be characteristic  of a `regular' first order phase transitions, as is found in the afore mentioned Q-state Potts models at $Q>4$.

It may be of relevance to note here that it has already been known that for the Ising model the transition is `continuous' in the sense that the magnetic susceptibility 
$\chi(\beta) := \sum_x \langle  \sigma_0 \sigma_x   \rangle_{\beta}  $, which  is finite for all $\beta < \beta_c$~(\cite{ABF87}), diverges as $\beta \nearrow \beta_c$~(by the argument of \cite{GJ74}).   That is now strengthened to the observation that $m^*(\beta_c)>0$ requires also $\lim_{\beta \nearrow \beta_c} M_{LRO} \ > \ 0$.

\item[iii.]  For the specific example of the one-dimensional Ising model with 
$J_{x,y}=1/|x-y|^2$ , for which  $m^*(\beta_c)>0$ (\cite{Tho69,ACCN88}), Theorem~\ref{thm:continuity} combined with the monotonicity of the two point function (\cite{Sch,MessMSole})
implies that for all $x\in \Z^d$ 
\be 
 \lim_{\beta \nearrow \beta_c}   \langle\sigma_0\sigma_x\rangle_{\beta}^0 \ >   \  m^*(\beta_c)^2 \ > \ 0 \,. 
\ee 
\end{enumerate} 
\medskip  

Theorem~\ref{thm:continuity} is  proven here by extending the random current representation of~\cite{Aiz82}, which is based on the  switching lemma  of~\cite{GHS70}, to  infinite domains and establishing uniqueness of the infinite cluster for the resulting system of (duplicated) random currents on $\bbZ^d$.   
This provides a useful tool for studying the implications of condition \eqref{eq:cont_cond}.

\subsection{Applications to reflection positive models} 

Among the few available tools for establishing the validity of  \eqref{eq:cont_cond} for nonsolvable models, and the only one which applies in the intermediate but important dimension $d=3$,  is the Gaussian domination bound of  \cite{FSS76,FILS}  which applies to 
reflection positive interactions (see e.g. \cite{Bis09} for a review on this crucial notion).  

For the Ising model on a torus, where due to the period boundary condition  the correlation function   $F_{L,\beta}(x,y) =  \langle \sigma_x\sigma_y\rangle_{\bbT_L,\beta} $ depends only on $x-y$,  we denote the Fourier transform by       
\be 
 \widehat F_{L,\beta}(p) \ :=\   \sum_{x\in \bbT_L}e^{i p \cdot x}F_{L,\beta}(0,x)  
 \ee 
where  $p$ ranges over $\bbT_L^\star=\left( \frac{2\pi}{L} \Z \right)^d \cap (-\pi,\pi]^d$, and  $ u\cdot v  $ denotes the  scalar product between $u$ and $v$.

\begin{proposition} [Gaussian domination bound~\cite{FSS76,FILS}]  
If the interaction $J_{x,y}$ is reflection positive, then or any $p\in\bbT_L^\star\setminus\{0\}$: 
\be 
\label{eq:89}  \widehat F_{L,\beta}(p)\le \frac{1}{2\beta E(p)} \, ,  \\[2ex]  
\ee
where 
\be
E(p)\ :=\  \sum_{x\in\bbZ^d}\left(1-e^{i p\cdot x} \right)J_{0,x} \ = 
\ 2 \sum_{x\in\bbZ^d} \sin^2\big(\tfrac{p\cdot x}{2}\big) \,  J_{0,x} \, .
\ee 
\end{proposition}

Here $E(p)$  is the   energy function of modes of momentum $p$.    The bound  \eqref{eq:89}  
allows to prove that \eqref{eq:cont_cond}  holds 
for  reflection-positive models provided   
\be  \label{eq:Ep_condition}
\int_{[-\pi,\pi]^d}\frac{{\rm d}p}{(2\pi)^d}\frac{1}{E(p)}\  < \  \infty ~.
\ee \\ 
The relation \eqref{eq:Ep_condition} is also the condition for transience of the random-walk associated with the weights  $(J_{x,y})_{x,y\in\bbZ^d}$, which  is the Markov process defined by the transition probabilities
\be 
\bbP(X_{n+1}=y|X_n=x)=\frac{J_{x,y}}{\sum_{z\in \bbZ^d}J_{x,z}  }\quad\text{ for all }x,y\in\bbZ^d.
\ee
This yields the following conclusion (derived below in Section~\ref{sec:3}): 

\begin{corollary}\label{cor:main}
If the random-walk associated to $(J_{x,y})_{x,y\in\bbZ^d}$ is transient and the model is reflection-positive, then the magnetization of the Ising model is continuous at $\beta_c$.
\end{corollary}

For specific applications, let us quote from~\cite{FSS76,FILS} (c.f. also \cite{AF86}) that the following Ising models are reflection-positive (with respect to hyperplanes passing through vertices):  
\vspace{-.2cm} 
\begin{enumerate} 
\item (nearest neighbor interactions) $J_{x,y} = \delta_{\|x-y\|_1,1}$\, , 
\item (exponential decay) \hspace{1.8cm}  $J_{x,y}=\exp(-\mu \|x-y\|_1)$ for $\mu>0$\, , 
\item (power-law potentials) \hspace{1.2 cm}  $J_{x,y}=\|x-y\|_1^{-\alpha}$ for $\alpha > d$\, ,  \\ 
\end{enumerate}
where  $\|x\|_1=\sum_{i=1}^d|x_i|$ for  $x=(x_1,\dots,x_d)$.  Furthermore, Ising models whose couplings are linear combinations with positive coefficients of the couplings mentioned above  are also reflection-positive,   
and in one dimension  the existence of a phase transition requires both $\sum_{n\in \mathbb N} J_n  = \infty$  and 
$\sum_{n\in \mathbb N} n \, J_n <  \infty$, which corresponds to $\alpha \in (1, 2]$  
(\cite{Dys69,ACCN88}).\\ 

In the above examples $E(p)$ vanishing for $p \to 0$ at the rates:  $E(p)\approx |p|^2$ in cases  (1) and (2), and  $E(p) \approx |p|^{\min\{2,\alpha-d\}}$ in case (3).   Thus, verifying the condition \eqref{eq:Ep_condition}  we conclude:  

\begin{corollary}\label{example}
The magnetization of the following models is continuous as a function of $\beta$ for
\vspace{-.4cm} 
\begin{itemize}[noitemsep]
\item any reflection positive ferromagnetic Ising model in $d>2$ dimensions,
\item any such one dimensional  model whose interaction  includes long range term(s) with $1< \alpha <2$, and no other powers.
\end{itemize}
\end{corollary}

It may be added that while  \eqref{eq:Ep_condition}  is not satisfied in two dimensions,  Theorem~\ref{thm:continuity} is of relevance  also for  this case,  since the condition \eqref{eq:cont_cond} can be established 
for the nearest neighbor model on the square, hexagonal and triangular lattices 
 through the Fortuin-Kasteleyn random cluster representation (see  \cite{Gri06} Theorem 6.72). 
\\ 

The  planar case   shows that   transience of the associated random walk (namely \eqref{eq:Ep_condition}) is not  necessary  for the continuity of $m^*(\beta)$.   On the other hand, for one-dimensional long range models with $J_{x,y} = 1/|x-y|^\alpha$ the criterion provided by \eqref{eq:Ep_condition} is sharp, since it holds just up to the value $\alpha=2$ at which the spontaneous magnetization is known to be discontinuous at $\beta_c$.  \\

\subsection{Past results on continuity at $\beta_c$} 

The  past results on the continuity of the spontaneous magnetization at $\beta_c$ are  naturally split  into two distinct classes: {\it i.} the special low dimensional case of  $d=2$, and {\it ii.} high dimensions (as described below).  For the standard nearest neighbor model, the only dimension which has been left out is the one of seemingly most physical interest:  $d=3$.   
Thus, aside from the general principle,  for the nearest-neighbor case the novelty in our results  is mainly limited to that case  (plus certain long range models that fall between these two ranges). 
 
The earliest results have been derived for  the   nearest neighbor model in $d=2$ dimensions, for which the spontaneous magnetization  was computed by  Yang~\cite{Yan52}, using  the methods of~\cite{Ons44,K49}.  Prior to that, Kaufman and Onsager~\cite{KO49} showed that $\widetilde M_{LRO}(\beta_c)=0$.   More recently, continuity of $m^*(\beta)$   at $\beta_c$ was given  a short proof in~\cite{Wer09} and a proof using discrete holomorphicity and the Russo-Seymour-Welsh theory was exposed in \cite{DHN10}. 

For high dimensions, the continuity of spontaneous magnetization in the nearest neighbor model was 
established in~\cite{AF86} for   $d\ge 4$ (formally $d > 3 \frac{1}{2}$)  through reflection positivity bounds combined with differential inequalities.   The method used there yields also information on the critical exponent $\delta$ with which $m^*(\beta) \approx |\beta_c-\beta|_{-}^{1/\delta}$, and  related results  for reflection positive long range interactions due to which the effective dimension (as expressed through the `bubble diagram') is lowered. \\   

Results which do not require reflection positivity were derived through the  general method of `lace expansion'~\cite{Sl06} whose adaptation to Ising systems' random currents was accomplished by A. Sakai~\cite{Sak07}.   
The lace expansion is restricted to models above the upper critical dimension, $d> d_c$, which in the presence of long range interaction (with $\alpha > d$) is lowered to $d_c = \min \{ 4, 2(\alpha -d)\}$  \cite{HHS08} (see also \cite{CS12}  for tight estimates on the two-point function).  
For optimal dimensional dependence of the results,  the lace expansion requires also   a sufficiently spread-out short range interaction (possibly for only  technical reasons).   However this method allows to deduce mean field behavior of $m^*(\beta)$ in high  dimensions for a collection of  non-solvable, and not necessarily reflection positive, models.   \\

\section{The Random Current representation and its percolation properties}

We shall use the following notation when discussing dependent percolation on the graph  $\Z^d$.    
\paragraph{Notation.}  For any subset $G \subset \Z^d$, we let  $\calP_2(G)=\{\{x,y\}:x,y\in G\}$. 
For a configuration of ``bond variables'' $\omega\in\{0,1\}^{\calP_2(\bbZ^d)}$, an edge $\{x,y\}$ for which $\omega_{x,y}=1$ is said to be {\em open}, and otherwise it is {\it closed}.  Two vertices $x$ and $y$ are said to be {\em connected} if there exist $x=x_0,\dots,x_m=y$ such that 
$\omega_{x_i,x_{i+1}}=1$ for every $0\le i<m$.   The statement that $x$ and $y$ are connected is denoted by $x\lr{\omega} y$ (and when $\omega$ is deemed clear from the context, we drop it from the notation).    The sites of $\bbZ^d$ are partitioned into maximal connected components of $\omega$, which are called {\em clusters}. 

We will often encounter also integer valued bond functions, i.e. elements   $\omega \in \{0,1,2,...\}^{\calP_2(G)} \ =: \ \Omega_G$.  The associated percolation would refer to  the projection $\widehat\cdot: \Omega_G \longrightarrow\{0,1\}^{\calP_2(G)}$ defined by 
$$\widehat{n}_{x,y}=\begin{cases}1&\text{ if ${n}_{x,y}>0$,}\\ 0&\text{ otherwise.}\end{cases}
$$

The `lattice shifts', by vectors $x\in \bbZ^d$, of configurations $\omega\in\{0,1\}^{\calP_2(\bbZ^d)}$, or $\omega \in \Omega_{\bbZ^d}$, are the  mappings $\tau_x$  defined by $\tau_x(\omega)_{a,b}=\omega_{a+x,b+x}$ 
for  all  $a,b\in\bbZ^d$.\\ 

The indicator function on a configuration space $\Omega$ corresponding to a condition $E$ will be denoted by $\indf{E}\equiv \indf{E}(\omega)$.  The argument $(\omega)$ will be omitted when its deemed to be clear within the context.   \\

\subsection{The random current representation}

\begin{definition}
A {\em current} $\n$ on $G\subset \bbZ^d$ (also called a current configuration) is a function from $\calP_2(G)$ to $\{0,1,2,...\}$.    A {\em source} of $\n=(\n_{x,y}:\{x,y\}\in \calP_2(G))$ is a vertex $x$ for which $\sum_{y\in G}{\n}_{x,y}$ is odd. The set of sources of $\n$ is denoted by $\partial\n$, and the  collection of current configurations on $G$ is  $\Omega_G$. 

\end{definition}

\subsubsection*{Random current representation for free boundary conditions}     

The partition function of a finite graph $G$ is: 
\begin{equation}\label{eq:4}Z^0(G,\beta)=\sum_{\sigma\in\{-1,1\}^G}\prod_{\{x,y\}\subset G}e^{\beta J_{x,y}\sigma_x\sigma_y}.\end{equation}
Expanding $e^{\beta J_{x,y}\sigma_x\sigma_y}$ for each $\{x,y\}$ into
$$e^{\beta J_{x,y}\sigma_x\sigma_y}=\sum_{{\n}_{x,y}=0}^\infty\frac{(\sigma_x\sigma_y)^{{\n}_{x,y}} (\beta J_{x,y})^{{\n}_{x,y}}}{{\n}_{x,y}!}$$
and substituting this relation in \eqref{eq:4}, one gets
\begin{equation*}
Z^0(G,\beta)=\sum_{\n\in\Omega_G}w_\beta(\n)\sum_{\sigma\in\{-1,1\}^G}\prod_{x\in G}\sigma_x^{\sum_{y\in G} {\n}_{x,y}},
\end{equation*}
where
$$w_\beta(\n):=\prod_{\{x,y\}\subset G}\frac{(\beta J_{x,y})^{{\n}_{x,y}}}{{\n}_{x,y}!}.$$
Now,
$$\sum_{\sigma\in\{-1,1\}^G}\prod_{x\in G}\sigma_x^{\sum_{y\in G} {\n}_{x,y}}=\begin{cases}0 &\text{ if $\sum_{y\in G} {\n}_{x,y}$ is odd for some $x\in G$,}\\
2^{|G|}&\text{ otherwise.}\end{cases}$$
Above, $|G|$ denotes the number of sites of $G$. 
Thus, the definition of a current's source  enables one to write
\be 
Z^0(G,\beta)=2^{|G|}\sum_{\n\in\Omega_G:\,\partial\n=\emptyset}w_\beta(\n)
~. 
\ee

Similar expansions for the correlation functions  involve  currents  with sources.   For instance,
\begin{equation*}\sum_{\sigma\in\{-1,1\}^G}\sigma_x\sigma_ye^{-\beta H^0_G(\sigma)}=2^{|G|}\sum_{\n\in\Omega_G:\,\partial\n=\{x,y\}}w_\beta(\n),\end{equation*}
which gives
\begin{equation}\label{eq:7}\langle\sigma_x\sigma_y\rangle_{G,\beta}^0=\frac{\displaystyle\sum_{\n\in\Omega_G:\,\partial\n=\{x,y\}}w_\beta(\n)}{\displaystyle\sum_{\n\in\Omega_G:\,\partial\n=\emptyset}w_\beta(\n)}.\end{equation}

\subsubsection*{Random current representation for $+$ boundary conditions} 
For a finite subset $G\subset \Z^d$, the equilibrium state at $+$ boundary conditions is obtained by freezing all spins in the complementary set to the value $+1$.   This may be conveniently represented by adding  an additional vertex $\delta\notin \bbZ^d$, to which we refer as the   {\em ghost spin} site\footnote{The notion is related to Griffith's  {\em ghost spin}, which was added by R.B. Griffiths  \cite{Grif_ghost}   as a tool for the extension of correlation inequalities to states under  an external field.},   
and setting the  coupling between it and sites $x\in G$   to  $J_{x,\delta}=J_{x,\delta}(G):=\sum_{y\notin G}J_{x,y}$.    

For notational convenience we adapt the convention that the ghost site is not to be listed in the configurations source set $\partial \n$ regardless of the  parity of the flux into $\delta$ (which can be determined from the parity of   $|\partial \n|$).\\

For $+$ boundary conditions, a development similar to the  above yields: 
\begin{equation}\label{eq:8}\langle\sigma_x\sigma_y\rangle_{G,\beta}^+=\frac{\displaystyle\sum_{\n\in\Omega_{G\cup\{\delta\}}:\,\partial\n=\{x,y\}}w_\beta(\n)}{\displaystyle\sum_{\n\in\Omega_{G\cup\{\delta\}}:\,\partial\n=\emptyset}w_\beta(\n)}\, . 
\end{equation}
Observe that \eqref{eq:7} differs from \eqref{eq:8} in that the summation is over all currents on $G\cup\{\delta\}$ instead of $G$. Also note that $J_{x,\delta}$ depends on $G$. \\ 

\subsubsection*{Switching lemma} 
As mentioned in the introduction, the  {\it random current} perspective on the Ising model's phase transition  is driven by the observation that the onset of long range order coincides with a percolation transition in a dual system of currents.   This point of  view,  as an intuitive guide to diagrammatic bounds  which under certain conditions provide `hard information' on the critical model's scaling limits,  
was developed in \cite{Aiz82} and a number of subsequent works.     
Among the first tools which facilitate cancellations in this representation is the following graph-theoretic switching lemma, which was originally introduced in \cite{GHS70} and applied there for the  Griffiths-Hurst-Sherman (GHS) inequality

\begin{lemma}[Switching lemma\footnote{In allowing $G$ to be a strict subset of $ H$, Lemma~\ref{switching} forms a minor extension of the statement found in \cite{GHS70}.   
An  allusion to it, and its other applications,  was made in the explanation of Lemma 6.3 in \cite{AG83}, where this extension was  applied.  
The difference in the proof  is rather trivial.
}]  \label{switching}
For any nested pair of finite sets $G\subset H$, pair of sites $x,y \in  G$ and   $A\subset H$, and a function
 $F:\Omega_H\rightarrow \bbR$: 
\begin{align*}\sum_{\substack{\n_1\in\Omega_G:\,\partial \n_1=\{x,y\}\\ \n_2\in\Omega_H:\,\partial \n_2=A }}&F(\n_1+\n_2)w_\beta(\n_1)w_\beta(\n_2)\\
&=\sum_{\substack{\n_1\in\Omega_G:\,\partial \n_1=\emptyset\\ \n_2\in\Omega_H:\,\partial \n_2=A\Delta\{x,y\}}}F(\n_1+\n_2)w_\beta(\n_1)w_\beta(\n_2)\indf{ x\lr{\widehat{\n_1+\n_2}}y\mathrm{\ in\ }G }.\end{align*}
with  $A\Delta B$ denoting the symmetric difference $(A\setminus B)\cup(B\setminus A)$ between $A$ and $B$.
\end{lemma}  
 
 The essential graph-theoretic argument is given in \cite{GHS70}, and in the random current notation which is employed below in \cite{Aiz82}.  Since the proof  provides  an introduction to the notation  let us repeat it here.  
\begin{proof}
In the argument, a current on  the subgraph corresponding to $G$ is also viewed as a current on $H$ which vanishes on pairs $\{ x,y \}$ not contained in  $G$.    The switching is performed within  collections of pairs of currents $\{ \n_1,\, \n_2 \}$ of a specified value for the sum $\m:=\n_1+\n_2$.   It is therefore convenient to take as the summation variables the current pairs $\m$  and $\n=\n_1 \le \m$  (with $\n \le \m$ defined as the natural  {\it partial order} relation).    One obtains
\begin{align*}\sum_{\substack{\n_1\in\Omega_G:\,\partial \n_1=\{x,y\}\\ \n_2\in\Omega_H:\,\partial \n_2=A }}&F(\n_1+\n_2)\, w_\beta(\n_1) \,  w_\beta(\n_2)\\
&=\sum_{\substack{\m\in\Omega_H:\,\partial \m=A\Delta\{x,y\}}} F(\m) \,  w_\beta(\m)\sum_{\substack{\n\in\Omega_G:\,\partial\n=\{x,y\}\\ \n\le\m}}\binom{\m}{\n},\end{align*}
and
\begin{align*}
\sum_{\substack{\n_1\in\Omega_G:\,\partial \n_1=\emptyset\\ \n_2\in\Omega_H:\,\partial \n_2=A\Delta\{x,y\} }}&F(\n_1+\n_2) \,  w_\beta(\n_1) \,  w_\beta(\n_2)\,  \indf { x\lr{\widehat{\n_1+\n_2}}y\text{ in }G }\\
&=\sum_{\substack{\m\in\Omega_H:\,\partial \m=A\Delta\{x,y\}}}F(\m) \,  w_\beta(\m)  \,  \indf{x\lr{\widehat\m}y\text{ in }G }\sum_{\substack{\n\in\Omega_G:\,\partial\n=\emptyset\\
\n\le\m}}\binom{\m}{\n},\end{align*}
where 
$\binom{\m}{\n}=\prod_{\{x,y\}\subset G}\binom{\m_{x,y}}{{\n}_{x,y}}$
 and where we used the fact that
$$w_\beta(\n_1)w_\beta(\n_2)=\prod_{\{x,y\}\subset G\cup\{\delta\}}\left[\frac{(\beta J_{x,y})^{{\n}_{x,y}}}{{\n}_{x,y}!}\right]\,  \left[\frac{(\beta J_{x,y})^{\m_{x,y}}}{\m_{x,y}!}\right]\  =\  w_\beta(\m)\binom{\m}{\n}.$$
The claim follows if the relation below is proved for every current $\m\in\Omega_H$:
\begin{equation}\label{eq:10}
\sum_{\substack{\n\in\Omega_G:\,\partial\n=\{x,y\}\\ \n\le\m}}\binom{\m}{\n}\ = \ \indf{ x\lr{\widehat\m}y\text{ in }G }\sum_{\substack{\n\in\Omega_G:\,\partial\n=\emptyset\\ \n\le\m}}\binom{\m}{\n}.
\end{equation}
First, assume that $x$ and $y$ are not connected in $G$ by $\m$. The right-hand side is trivially zero. Moreover, there is no current $\n$ on $G$ which is smaller than $\m$ and which connects $x$ to $y$. The left-hand side is thus 0 and \eqref{eq:10} is proved in this case.

Let us now assume that $x$ and $y$ are connected in $G$ by $\m$. Associate to $\m$ the graph $\calM$ with vertex set $G$, and $\m_{a,b}$ edges between $a$ and $b$. For a subgraph $\calN$ of $\calM$, let $\partial\calN$ be the set of vertices belonging to an odd number of edges. Since $x$ and $y$ are connected in $G$ by $\m$, there exists a subgraph $\calK$ of $\calM$ with $\partial\calK=\{x,y\}$.

The involution $\calN\mapsto \calN\Delta\calK$ provides a bijection between the set of subgraphs of $\calM$ with $\partial\calN=\emptyset$, and the set of subgraphs of $\calM$ with $\partial\calN=\{x,y\}$. Therefore, these two sets have the same cardinality. Since the summations in 
$$\sum_{\substack{\n\in\Omega_G:\,\partial\n=\emptyset\\ \n\le\m}}\binom{\m}{\n}\quad\text{and}\quad\sum_{\substack{n\in\Omega_G:\,\partial\n=\{x,y\}\\ \n\le\m}}\binom{\m}{\n}$$
are over currents $\n$ in $G$, these sums correspond to the cardinality of the two sets mentioned above. In particular, they are equal and the statement follows.
\end{proof}
\subsection{Infinite-volume random current representation}

Next, we formulate the infinite volume limit of the random current representation.   This allows a more effective use of the asymptotic translation invariance, and enables us to  deploy the uniqueness of the infinite cluster argument, which can be established in this   context.    

Let ${\rm P}^0_{G,\beta}$ be the law on currents on $G$ defined by
\be  {\rm P}^0_{G,\beta}[\n]:=\frac{w_\beta(\n) \, \,  \indf{\partial\n=\emptyset} }{\displaystyle\sum_{\m\in\Omega_G:\,\partial\m=\emptyset}w_\beta(\m)}\quad, \qquad  
\forall \n\in\Omega_G~. 
\ee 
This  induces a measure on $\Omega_G:=\{0,1,2,...\}^{\calP_2(G)}$ that we  denote by $\widehat {\rm P}^0_{G,\beta}$. One may also define a law on currents on $G\cup\{\delta\}$ which induces a measure on 
$\Omega_{G\cup\{\delta\}} :=\{0,1,2,...\}^{\calP_2(G\cup\{\delta\})}$
denoted by $\widehat {\rm P}^+_{G,\beta}$. 

Let $\Lambda_L=[-L,L]^d$ be the box of size $L$.

\begin{theorem}\label{def:current infinite}
Let $\beta>0$. There exist two laws $\widehat{\rm P}^+_\beta$ and $\widehat{\rm P}^0_\beta$ on 
$\Omega_{\bbZ^d}$  such that
\begin{itemize}
\item[\rm \bf R1] (Convergence)  For any event $\calA$ depending on finitely many edges, 
$$\lim_{L\rightarrow \infty}\widehat{\rm P}^+_{\Lambda_L,\beta}[\calA]=\widehat{\rm P}^+_\beta[\calA]\quad\text{and}\quad \lim_{L\rightarrow \infty}\widehat{\rm P}^0_{\Lambda_L,\beta}[\calA]=\widehat{\rm P}^0_\beta[\calA].$$
\item[\rm \bf R2] (Invariance under translations) $\widehat{\rm P}^+_\beta$ and $\widehat{\rm P}^0_\beta$ are invariant under the shifts $\tau_x$, $x\in \bbZ^d$.
\item[\rm \bf R3] (Ergodicity) $\widehat{\rm P}^+_\beta$ and $\widehat{\rm P}^0_\beta$ are ergodic with respect to the group of shifts $(\tau_x)_{x\in\bbZ^d}$.
\end{itemize}\end{theorem}

\begin{proof}
Except for a minor difference in the very last step  the proof is identical for the $(+)$ and the free $(f)$ boundary  conditions.   Let us therefore use the symbol $\#$ as a marker for either of the two.    
 
\paragraph{Proof of R1} (Convergence) 

To prove convergence of the finite volume probability measures, let us first note that the distribution of the random currents simplifies into a product measure when conditioned on the parity variables ${\bf r}(\omega) = ({\bf r}_{x,y})_{x,y \subset G}$, with:  
\be 
{\bf r}_{x,y} (\omega)\ := \  (-1)^{{\bf n}_{x,y} (\omega)}  \, .
\ee  
The conditional distribution of $\n$, given ${\bf r}(\omega)$, is simply the product measure of independent Poisson processes of mean values $\beta J_{x,y}$ conditioned on the corresponding parity.   Thus, for a proof of convergence it suffices to establish convergence of the law of the parity variables ${\bf r}(\omega)$.   

For a set of bonds (i.e. graph edges) $E\subset\calP_2(\bbZ^d)$, define the events 
\begin{eqnarray}  
\calC_E  & =&  \left \{ \ \omega \ : \  {\bf r}_{x,y}(\omega)  = 1  \quad \forall \{x,y\} \in E \  \right \} \, ,
\\
\calC^{(0)}_E  & =&  \left \{ \ \omega \ : \  {\bf n}_{x,y}(\omega)  = 0  \quad \forall \{x,y\} \in E \  \right \}\,  . \notag  
\end{eqnarray} 
Let us prove that for any finite subset $E$ of edges of $\bbZ^d$, $\widehat{\rm P}^\#_{\Lambda_L,\beta}[\calC_E]$ converges as $L$ tends to infinity. 

To facilitate a unified treatment of the two boundary conditions we denote 
\be 
\Omega^\#_L \ = \ 
\begin{cases}  
\Omega_{\Lambda_L  }& \mbox{for $\# = 0$} \, ,  \\[2ex]   
\Omega_{\Lambda_L\cup\{\delta\} }& \mbox{for $\# = +$} \, . 
\end{cases} \, 
\ee 

For $L$ large enough (so that $\Lambda_L \supset E$) we have: 
\begin{eqnarray} \label{eq:deletion}
  \widehat{\rm P}^\#_{\Lambda_L,\beta}[\calC_E] \  & = & \frac{\displaystyle\sum_{\n\in\Omega^\#_L:\,\partial \n=\emptyset}w_\beta(\n)\, \indf{\calC_E}}{\displaystyle\sum_{\n\in\Omega^\#_L:\,\partial \n=\emptyset}w_\beta(\n) }   \notag \\ 
 & = & 
  \frac{\displaystyle\sum_{\n\in\Omega^\#_L:\,\partial \n=\emptyset}w_\beta(\n)\, \indf{\calC^{(0)}_E}}{\displaystyle\sum_{\n\in\Omega^\#_L:\,\partial \n=\emptyset}w_\beta(\n) }  \,  \prod_{x,y\in E} \cosh (\beta J_{x,y}) 
  \notag  \\[2ex]  
  & = &
     \frac{ Z^\#(\Lambda_L \setminus  E, \beta) } { Z^\#(\Lambda_L,\beta)}  \,  \prod_{x,y\in E} \cosh (\beta J_{x,y}) \, 
  \end{eqnarray}
Above, $\Lambda_L\setminus E$ designates the graph obtained by removing the edges of $E$ but keeping all the vertices of $\Lambda_L$. 
The above ratio can be expressed in terms of an expectation value of a finite term: 
\be  \label{eq:ProbE}
\widehat{\rm P}^\#_{\Lambda_L,\beta}[\calC_E] \ 
= \      \big \langle e^{  -\beta K_E   } \big \rangle^\#_{\Lambda_L, \beta} \prod_{x,y\in E} \cosh (\beta J_{x,y})\ee 
with the finite volume collection of energy terms 
\be \label{eq:K} 
K_E(\omega) := \sum_{x,y\in E}    J_{x,y} \sigma_x \sigma_y \, . 
\ee
The convergence of the above expression  follows now directly from the convergence of correlation functions as $L$ tends to infinity.

The events $\calC_E$ with $E$ ranging over finite sets of edges span (by inclusion-exclusion) the algebra of events expressible in terms of finite collections of the binary variables of ${\bf r}(\omega)$.   This fact, and the above observation that the  probability distribution of the random current $\n$ conditioned on ${\bf r}(\omega)$  does not depend on $L$, implies the existence of $\widehat{\rm P}^\#_\beta$. 

\paragraph{Proof of R2} (Translation invariance) 

 Fix $x\in \bbZ^d$. The limit of the probability of the event $\calC_E$, where $E$ is a finite set of edges, is the same if the sequence $(\Lambda_L)_{L\ge0}$ is replaced by the sequence $(x+\Lambda_L)_{L\ge 0}$. (Simply use \eqref{eq:ProbE} and the convergence of $\langle \cdots \rangle_{{x+\Lambda_L},\beta}^\#$ to $\langle \cdots \rangle_{\beta}^\#$.) This immediately implies that $\widehat{\rm P}^\#_{\beta}$ is invariant under translations.

\paragraph{Proof of R3} (Ergodicity)  

Since every translationally invariant event can be approximated by events depending on a finite number of edges, it is sufficient to prove that for any events $A$ and $B$ depending on a finite number of edges,
\be  \label{eq:ergcond}
\lim_{\|x\|_1\rightarrow \infty}\widehat{\rm P}^\#_\beta[A\cap \tau_xB]\ =\  \widehat{\rm P}^\#_\beta[A]\  \widehat{\rm P}^\#_\beta[B]. 
\ee 
In view of the conditional independence of $\n$ given the parity variables ${\bf r}$, the requirement can be further simplify to the proof that for any two finite sets $E$ and $F$ of edges, 
\be 
\lim_{\|x\|_1\rightarrow\infty}\widehat{\rm P}^\#_\beta[\calC_{E\cup(x+F)}]\  =\  \widehat{\rm P}^\#_\beta[\calC_E] \  \widehat{\rm P}^\#_\beta[\calC_F] ~. 
\ee 
Using the  expression \eqref{eq:ProbE}, for $x$ large enough so that  $E\cap (x+F) =\emptyset$: 
\begin{eqnarray}   \label{eq:ProbE2}
\frac {  \widehat{\rm P}^\#_{\beta}[\calC_{E\cup(x+F)}]  } 
       {\widehat{\rm P}^\#_\beta[\calC_E] \  \widehat{\rm P}^\#_\beta[\calC_F] } &  
= &    
  \frac {\big \langle e^{   -\beta K_E } e^{   -\beta K_{x+F} } \big \rangle^\#_{ \beta } } {  \big \langle e^{-\beta K_E } \big \rangle^\#_{\beta }  \   \big \langle e^{-\beta K_F } \big \rangle^\#_{\beta }  } 
\end{eqnarray}
Ergodicity of the random current states can therefore be presented as an implication of the  statement that this ratio tends to 1.   
This condition holds as a consequence of the mixing property of the states 
$ \big \langle   \cdots     \big \rangle^{\#}_\beta$  when restricted to  functions which are invariant under global spin flip (i.e. that $f(-\sigma)=f(\sigma)$ for every spin configuration $\sigma$).   For completeness we enclose the proof of the statement, which may be part of the folklore among experts, in Appendix A. 
\end{proof}

\begin{remark}  The relation of the ergodicity of $\widehat{\rm P}_\beta^+$ and $\widehat{\rm P}_\beta^0$ to the partial ergodicity of $\langle \cdots \rangle_\beta^+$ or $\langle \cdots \rangle_\beta^0$ (i.e. ergodicity of only the restriction to the $\sigma$ algebra of even event) can be compared to a similar relation in the random-cluster representation of the $Q$-state Potts models, with wired and free boundary conditions (see \cite{Gri06} for more details on these models).    The restriction is needed since  in the presence of symmetry breaking (at $\beta>\beta_c$) the state    $\langle \cdots \rangle_\beta^0$  is  not even mixing on functions which are odd with respect to the global spin flip. 

\end{remark}

\subsection{Percolation properties of the sum of random currents}
Define 
$\mathbb P_\beta$ to be the law of $\widehat{\n_1+\n_2}$, where  $\n_1$ and $\n_2$ are two independent currents with laws ${\rm P}^0_{\beta}$ and ${\rm P}^+_{\beta}$. We also set $\mathbb E_\beta$ for the expectation with respect to $\mathbb P_\beta$. 
Properties {\bf R2} and {\bf R3} of Theorem~\ref{def:current infinite} imply immediately that $\mathbb P_\beta$ is invariant and ergodic with respect to shifts.

We now prove that there cannot be more than one infinite cluster. This claim will be crucial in the proof of Theorem~\ref{thm:continuity}: it will replace the use of the FKG inequality, which is not available for the random current representation.
\begin{theorem}\label{thm:percolation}
For any translation invariant ferromagnetic  Ising model on $\Z^d$ satisfying $\mathbf{C1}-\mathbf{C4}$, there exists at most one infinite cluster $\mathbb P_\beta$-almost surely (at any $\beta \ge 0$).
\end{theorem}

The following lemma is an equivalent of the insertion tolerance valid for many spin models. 

\begin{lemma}\label{lem:insertion}
Let $\widehat\Phi_N: \{0,1\}^{\calP_2(\bbZ^d)}\longrightarrow \{0,1\}^{\calP_2(\bbZ^d)}$ be the map opening all edges $\{x,y\}$ in $\Lambda_{N}$ with $J_{x,y}>0$. Let $N>0$, then there exists $c=c(N,J,\beta)>0$ such that for any event $\calE$,
$$\mathbb P_\beta[\widehat\Phi_N(\calE)] \ \ge\  c\  \mathbb P_\beta[\calE].$$ 
\end{lemma}

\begin{proof}
It is sufficient to consider events $\calE$ depending on a finite number of edges. Let $\mathbb P_{\Lambda_n,\beta}$ be the law of $\widehat{\n_1+\n_2}$, where $\n_1$ and $\n_2$ are two independent currents with respective laws ${\rm P}_{\Lambda_n,\beta}^0$ and ${\rm P}_{\Lambda_n,\beta}^+$.
Property {\bf R1} of Theorem~\ref{def:current infinite} shows that $\mathbb P_{\Lambda_n,\beta}$ converges weakly to $\mathbb P_\beta$. This reduces the proof to showing the existence of $c=c(N,J,\beta)>0$ on $\Lambda_n$, with a value which does not depend on $n>N$. 

Consider the transformation $\Phi_N:(\Omega_{\Lambda_n})^2\longrightarrow(\Omega_{\Lambda_n})^2$ defined by
$$\Phi_N(\n_1,\n_2)\{x,y\}=\begin{cases}(0,2)&\text{ if $(\n_1\{x,y\},\n_2\{x,y\})=(0,0)$,} \\
&\text{ $J_{x,y}>0$, and $x,y\in\Lambda_{N}$,}\\
(\n_1\{x,y\},\n_2\{x,y\})&\text{ otherwise},\end{cases}$$  
where exceptionally $\m\{x,y\}$ denotes $\m_{\{x,y\}}$ for ease of notation.
Fix $\omega\in\{0,1\}^{\Lambda_n}$ and let $\Omega_2=\{(\n_1,\n_2)\in(\Omega_{\Lambda_n})^2:\widehat{\n_1+\n_2}=\omega\}$. The set $\Phi_N(\Omega_2)$ is obtained from $\Omega_2$ by changing the value of the current $\n_2$ on edges $\{x,y\}$ with $\omega_{\{x,y\}}=0$ from 0 to 2. 
Therefore,
\begin{multline}\mathbb P_{\Lambda_n,\beta}[\widehat\Phi_N(\calE)] \ =\  \sum_{\omega'\in\widehat\Phi(\calE)}\mathbb P_{\Lambda_n,\beta}[\omega']=\sum_{\omega\in\calE}\frac{1}{{\rm Card}[\widehat\Phi_N^{-1}(\widehat\Phi_N(\omega))]}\mathbb P_{\Lambda_n,\beta}[\widehat\Phi_N(\omega)]  
 \\[2ex]   
\ge \  2^{-{\rm Card}(\calP_2(\Lambda_N))}\sum_{\omega\in\calE}\mathbb P_{\Lambda_n,\beta}[\widehat\Phi_N(\omega)] \ =\  
2^{-{\rm Card}(\calP_2(\Lambda_N))}\sum_{\omega\in\calE}{\rm P}_{\Lambda_n,\beta}^0\otimes{\rm P}_{\Lambda_n,\beta}^+[\Phi_N(\Omega_2)] 
\\[2ex]  
\ge\  2^{-{\rm Card}(\calP_2(\Lambda_N))}\sum_{\omega\in\calE}\Big(\prod_{\{x,y\}\subset\Lambda_N:\,\omega_{x,y}=0}\frac{(\beta J_{x,y})^2}{2}\Big){\rm P}_{\Lambda_n,\beta}^0\otimes{\rm P}_{\Lambda_n,\beta}^+[\Omega_2]  \\[2ex]  
\mbox{ } \qquad \quad \ge \  c\  \sum_{\omega\in\calE}{\rm P}_{\Lambda_n,\beta}^0\otimes{\rm P}_{\Lambda_n,\beta}^+[\Omega_2]\ =\ c\ \mathbb P_{\Lambda_n,\beta}[\calE] 
\, ,  \hfill \end{multline}
where $c=c(N,J,\beta)>0$ does not depend on $n$. 
In the first inequality, we used the fact that the number of pre-images of each configuration is smaller than 2 to the power the number of pairs of points in $\Lambda_{N}$ (since one has to decide whether edges of $\Lambda_N$ were open or closed before the transformation). 
\end{proof}

\begin{proof}[Proof of Theorem~\ref{thm:percolation}] 
For $\ell\in\bbN\cup\{\infty\}$, let $\calE_\ell$ be the event that there exist exactly $\ell$ disjoint infinite clusters. We must prove that $\mathbb P_\beta[\calE_\ell]=0$ for $\ell\ge 2$. The proof is based on a variation of the Burton-Keane argument \cite{BK89}. Let $k>0$ such that for any vertex $y$ satisfying $\|y\|_1=1$, there exist $0=x_0,\dots,x_m=y$ with $J_{x_0,x_1}\dots J_{x_{m-1},x_m}>0$ and with $x_i\in\Lambda_k$ for every $i\le m$ (the existence of this $k$ is guaranteed by the aperiodicity condition).

\paragraph{Proof of $\mathbb P_\beta[\calE_\ell]=0$ for $2\le \ell<\infty$.} Let $\ell\ge 2$.  Let $\calF_n$ be the event that the $\ell$ infinite clusters intersect $\Lambda_n$. 
Fix $N>k$ large enough so that $\mathbb P_\beta[\calF_N]\ge\tfrac12\mathbb P_\beta[\calE_\ell]$.
Lemma~\ref{lem:insertion} implies that
$\mathbb P_\beta[\widehat\Phi_{2N}(\calF_N)]\ge \tfrac c2\mathbb P_\beta[\calE_\ell]$.
Any configuration in $\widehat\Phi_{2N}(\calF_N)$ contains exactly one infinite cluster since all the vertices in $\Lambda_N$ are connected. Therefore, 
$$\mathbb P_\beta[\calE_1]\ge \tfrac c2\mathbb P_\beta[\calE_\ell].$$
Ergodicity implies that $\mathbb P_\beta[\calE_\ell]$ and $\mathbb P_\beta[\calE_1]$ are equal to $0$ or $1$, therefore $\mathbb P_\beta[\calE_\ell]=0$.

\paragraph{Proof of $\mathbb P_\beta[\calE_\infty]=0$.} Assume that $\mathbb P_\beta[\calE_\infty]>0$ and consider $N>2k$ large enough so that 
$$\mathbb P_\beta[\text{three distinct infinite clusters intersect the box $\Lambda_{N/2}$}]>0.$$
Lemma~\ref{lem:insertion} (applied to $\widehat\Phi_{N}$) implies that  $\mathbb P_\beta[{\rm CT}_0]>0$, where ${\rm CT}_0$ is the following event:   
\begin{itemize}[noitemsep]    
\item all vertices in $\Lambda_{N/2}$ are connected to each other in $\Lambda_N$,
\item If $\calC$ is the cluster of 0, then $\calC\cap(\bbZ^d\setminus\Lambda_N)$ contains at least three distinct infinite connected components.\end{itemize}
A vertex $x\in (2N+1)\bbZ^d$ is called a coarse-trifurcation if $\tau_x{\rm CT}_0=:{\rm CT}_x$ occurs. By invariance under translation, $\mathbb P_\beta[{\rm CT}_x]=\mathbb P_\beta[{\rm CT}_0]$.
     
Fix $n\gg N$. The set $T$ of {\em coarse-trifurcations} in $\Lambda_n$ has a natural structure of forest $\calF$ constructed inductively as follows. 
 \begin{itemize}     
 \item[Step 1] At time 0, all the vertices in $T$ are {\em unexplored}. 
 
 \item[Step 2] If there does not exist any unexplored vertex in $T$ left, the algorithm terminates. Otherwise, pick an unexplored vertex $t\in T$ and mark it explored (by this we mean that it is not considered as an unexplored vertex anymore). Go to Step 3. 
 
 \item[Step 3] Consider the cluster $\calC_t$ of vertices $x\in \bbZ^d$ connected to $t$ in $\Lambda_n$. This cluster decomposes into $k\ge 3$ disjoint connected components of $\bbZ^d\setminus (t+\Lambda_N)$ denoted $\calC_t^{(1)},\dots,\calC_t^{(k)}$.   For $i=1,\dots, k$, do the following:
  \begin{itemize}    
 \item if there exist two vertices $x\in\calC_t^{(i)}$ and $y\notin \Lambda_n$ such that 
  $\{x,y\}$ is open, and there exists an open path in $\Lambda_n$ going from $x$ to $t$ and not passing at distance $N$ from a coarse-trifurcation in $\calC_t^{(i)}\cap (T\setminus\{t\})$,
 then add the vertex $y$ to $\calF$ together with the edge $\{t,y\}$.  
 \item if there is no such vertex, then there must be a coarse-trifurcation $s\in\calC_t^{(i)}$ connected by an open path not passing at distance $N$ from a trifurcation in $\calC_t^{(i)}\cap (T\setminus\{t,s\})$. If $s$ is not already a vertex of $\calF$, add it. Then, add the edge $\{t,s\}$. \end{itemize}
 \item[Step 4] Go to Step 2.
 \end{itemize}
 The graph obtained is a forest (due to the structure of coarse-trifurcations). Each coarse-trifurcation corresponds to a vertex of the forest of degree at least three. Thus, the number of coarse-trifurcations must be smaller than the number of leaves.   Let $N$ be the number of leaves, we find
 \begin{equation}\label{eq:198}\mathbb P_\beta[{\rm CT}_0]\frac{(2n+1)^d}{(2N+1)^d}\le \mathbb E_\beta[N].\end{equation}
Yet, leaves are vertices outside $\Lambda_n$ which are connected by an open edge to a vertex in $\Lambda_n$, therefore
$$\mathbb E_\beta[N]\le 2d(2n+1)^{d-1}\sum_{k=0}^n \sum_{x \in \bbZ^d: \, \|x\|_1 \ge k}J_{0,x} ~.$$
Since $|J|<\infty$, we find that
$$0<\frac {\mathbb P_\beta[{\rm CT}_0]}{(2N+1)^d}\le \frac{\mathbb E_\beta[N]}{(2n+1)^d}\longrightarrow 0\quad\text{as $n\rightarrow \infty$. }$$
This contradicts $\mathbb P_\beta[{\rm CT}_0]>0$ and therefore $\mathbb P_\beta[\calE_\infty]$ must be zero. The claim follows.\end{proof}

\section{Proofs of  the main results} 

\subsection{A bound on the percolation probability} 

 Let us start with a crucial relation which justifies the consideration of $\mathbb P_\beta$.   
\begin{theorem}\label{thm:percolation}
For  $\beta$ at which $\widetilde M_{LRO}(\beta) =0$, also  $\mathbb P_{\beta}\left[0\leftrightarrow\infty\right] =  0.$
\end{theorem}

\begin{proof}
Let $L>0$ and let $x,y\in\Lambda_L$. The switching lemma (Lemma~\ref{switching}) implies
\begin{multline}
\label{eq:2}{\rm P}^0_{\Lambda_L,\beta}\otimes {\rm P}^+_{\Lambda_L,\beta}\left[x\lr{\widehat{\n_1+\n_2}}y\text{ in }\Lambda_L\right] \\[1ex]  
:=\frac{\displaystyle\sum_{\substack{\n_1\in\Omega_{\Lambda_L}:\,\partial \n_1=\emptyset\\ \n_2\in\Omega_{\Lambda_L\cup\{\delta\}}:\,\partial \n_2=\emptyset}}w_\beta(\n_1)w_\beta(\n_2)\indf{x\lr{\widehat{\n_1+\n_2}}y\text{ in }\Lambda_L}}{\displaystyle\sum_{\substack{\n_1\in\Omega_{\Lambda_L}:\,\partial \n_1=\emptyset\\ \n_2\in\Omega_{\Lambda_L\cup\{\delta\}}:\,\partial \n_2=\emptyset}}w_\beta(\n_1)w_\beta(\n_2)}   \notag \\[2ex]  
\end{multline} 
\begin{align}
&=\frac{\displaystyle\sum_{\substack{\n_1\in\Omega_{\Lambda_L}:\,\partial \n_1=\{x,y\}\\ \n_2\in\Omega_{\Lambda_L\cup\{\delta\}}:\,\partial \n_2=\{x,y\} }}w_\beta(\n_1)w_\beta(\n_2)}{\displaystyle\sum_{\substack{\n_1\in\Omega_{\Lambda_L}:\,\partial \n_1=\emptyset\\ \n_2\in\Omega_{\Lambda_L\cup\{\delta\}}:\,\partial \n_2=\emptyset}}w_\beta(\n_1)w_\beta(\n_2)}.\end{align}
The representations of spin-spin correlations \eqref{eq:7} and \eqref{eq:8} then imply that
\begin{equation}\label{eq:6}{\rm P}^0_{\Lambda_L,\beta}\otimes {\rm P}^+_{\Lambda_L,\beta}\left[x\lr{\widehat{\n_1+\n_2}}y\text{ in }\Lambda_L\right]=\langle\sigma_x\sigma_y\rangle_{\Lambda_L,\beta}^0\langle\sigma_x\sigma_y\rangle_{\Lambda_L,\beta}^+\le \langle\sigma_x\sigma_y\rangle_{\Lambda_L,\beta}^0.
\end{equation}
The right-hand side converges to $\langle\sigma_x\sigma_y\rangle_{\beta}^0$ as $L$ tends to infinity. Since the event on the left-hand side can be expressed in terms of $\widehat\n_1$ and $\widehat\n_2$, the convergence of $\widehat{\rm P}^+_{\Lambda_L,\beta}$ and $\widehat{\rm P}^0_{\Lambda_L,\beta}$ to $\widehat{\rm P}^+_{\beta}$ and $\widehat{\rm P}^0_{\beta}$ provided by Theorem~\ref{def:current infinite} implies that the left-hand side converges to $\mathbb P_{\beta}\left[x\leftrightarrow y\right]$. (The percolation event does not depend on finitely many edges, but justifying passing to the limit is straightforward by first considering the events that 0 is connected to distance $N$.) Therefore,
\begin{equation}\label{eq:111}
\mathbb P_{\beta}\left[x\leftrightarrow y\right]\ \le \  \langle\sigma_x\sigma_y\rangle_{\beta}^0\,. 
\end{equation}
Let now $B \subset \Z^d$ be a finite subset. The Cauchy-Schwarz inequality applied to the random variable $X=\sum_{x\in B} \indf{ x\leftrightarrow \infty }$ leads to
\begin{align*} 
\big(\,|B| \,  \mathbb P_{\beta}[0\leftrightarrow\infty]\,\big)^2 :=\mathbb E_\beta[X]^2&\le \mathbb E_\beta[X^2]=:\sum_{x,y\in B}\mathbb P_{\beta}[x,y\leftrightarrow\infty]. 
\end{align*}
The uniqueness of the infinite cluster thus implies
\begin{align}     \label{eq:133}
     \big(\,|B| \, \mathbb P_{\beta}[0\leftrightarrow\infty]\,\big)^2 &\le \sum_{x,y\in B}\mathbb P_{\beta}[x,y\leftrightarrow\infty]\le\sum_{x,y\in B}\mathbb P_{\beta}[x\leftrightarrow y].\end{align}
Combining this relation with \eqref{eq:111} and optimizing over $B$ we get: 
\begin{align}   \label{eq:567}  
  \mathbb  P_{\beta}[0\leftrightarrow\infty ]\,  ^2 \ & 
    \le \ \inf_{B\in \Z^d, |B|< \infty}  \frac{1}{|B|^2}\sum_{x,y\in  B}\langle\sigma_x\sigma_y\rangle_{\beta}^0  
     \  =  \  \widetilde M_{LRO}(\beta)^2~,  
\end{align}
which proves the claim.
\end{proof}

\begin{remark}   In the last step of \eqref{eq:133} one can see uniqueness of the infinite cluster used as a substitute for the classical percolation argument which utilizes the FKG inequality, which we do not have for random currents.\end{remark}

\subsection{Proof of Theorem~\ref{thm:continuity} }

Fix a pair of vertices $x$ and $y$. Applying the switching lemma  (Lemma~\ref{switching}) again, we find that for $L>0$:

\begin{align}&\langle\sigma_x\sigma_y\rangle_{\Lambda_L,\beta}^+-\langle\sigma_x\sigma_y\rangle_{\Lambda_L,\beta}^0=\frac{\displaystyle\sum_{\n_2\in\Omega_{\Lambda_L\cup\{\delta\}}:\,\partial \n_2=\{x,y\}}w_\beta(\n_2)}{\displaystyle\sum_{\n_2\in\Omega_{\Lambda_L\cup\{\delta\}}:\,\partial \n_2=\emptyset}w_\beta(\n_2)}-\frac{\displaystyle\sum_{\n_1\in\Omega_{\Lambda_L}:\,\partial \n_1=\{x,y\}}w_\beta(\n_1)}{\displaystyle\sum_{\n_1\in\Omega_{\Lambda_L}:\,\partial \n_1=\emptyset}w_\beta(\n_1)} \notag \\[2ex] 
&=\frac{\displaystyle\sum_{\substack{\n_1\in\Omega_{\Lambda_L}:\,\partial \n_1=\emptyset\\ \n_2\in\Omega_{\Lambda_L\cup\{\delta\}}:\,\partial \n_2=\{x,y\}}}w_\beta(\n_1)w_\beta(\n_2)-\sum_{\substack{\n_1\in\Omega_{\Lambda_L}:\,\partial \n_1=\{x,y\}\\ \n_2\in\Omega_{\Lambda_L\cup\{\delta\}}:\,\partial \n_2=\emptyset}}w_\beta(\n_1)w_\beta(\n_2)}{\displaystyle\sum_{\substack{\n_1\in\Omega_{\Lambda_L}:\,\partial \n_1=\emptyset\\ \n_2\in\Omega_{\Lambda_L\cup\{\delta\}}:\,\partial \n_2=\emptyset}}w_\beta(\n_1)w_\beta(\n_2)}  \notag   
\end{align} 
\begin{align} 
&=\frac{\displaystyle\sum_{\substack{\n_1\in\Omega_{\Lambda_L}:\,\partial \n_1=\emptyset\\ \n_2\in\Omega_{\Lambda_L\cup\{\delta\}}:\,\partial \n_2=\{x,y\}}}w_\beta(\n_1)w_\beta(\n_2)\left(1-\indf{x\lr{\widehat{\n_1+\n_2}}y\text{ in }\Lambda_L}\right)}{\displaystyle\sum_{\substack{\n_1\in\Omega_{\Lambda_L}:\,\partial \n_1=\emptyset\\ \n_2\in\Omega_{\Lambda_L\cup\{\delta\}}:\,\partial \n_2=\emptyset}}w_\beta(\n_1)w_\beta(\n_2)}\,. 
\end{align}
Yet, any configuration $\n_2$ with sources at $x$ and $y$ such that  $x$ and $y$ are not connected in $\Lambda_L$ necessarily satisfies that  $x$ and $y$ are connected to $\delta$. Therefore,
\begin{equation}\label{eq:57} 
\langle\sigma_x\sigma_y\rangle_{\Lambda_L,\beta}^+-\langle\sigma_x\sigma_y\rangle_{\Lambda_L,\beta}^0\le\frac{\displaystyle\sum_{\substack{\n_1\in\Omega_{\Lambda_L}:\,\partial \n_1=\emptyset\\ \n_2\in\Omega_{\Lambda_L\cup\{\delta\}}:\,\partial \n_2=\{x,y\}}}w_\beta(\n_1)w_\beta(\n_2) \indf{  x\lr{\widehat{\n_1+\n_2}}\delta }}{\displaystyle\sum_{\substack{\n_1\in\Omega_{\Lambda_L}:\,\partial \n_1=\emptyset\\ \n_2\in\Omega_{\Lambda_L\cup\{\delta\}}:\,\partial \n_2=\emptyset}}w_\beta(\n_1)w_\beta(\n_2)}.\end{equation} 
We shall now estimate the sum in the numerator by comparing it to the corresponding sum in which the source condition of $(\n_1,\n_2)$ is changed to   
$\partial \n_1 = \partial \n_2 =\emptyset $.   

Fix a sequence of vertices $x=x_0,\dots,x_m=y$ with $J_{x_i,x_{i+1}}>0$ for any $0\le i<m$.   For any $L$ large enough so that $x_i\in\Lambda_L$ for all $i\le m$, consider the one-to-many mapping which  assigns to  each $\omega$ a modified current configuration $\n_2$  with the change limited to $\n_2$ along the set of bonds $e_j= \{ x_i, x_{i+1}\}$,  $j=0,\ldots,m-1$,  at which the parity of all these variables is flipped, and the value of the new one is at least $1$ at each bond.    Under this mapping, the image of each pair $ (\n_1,\n_2)$ that contributes in the numerator of \eqref{eq:57} lies in the set for which  the connection event $x\lr{\widehat{\n_1+\n_2}}\delta $ remains satisfied, but the source set of $\n_2$ is reset to
$\partial \n_2 =\emptyset $.

Classifying the  current pairs according to the values of all the unaffected variables of $\{ (\n_1,\n_2)\}$, and the parity of $\n_2$ along the set of bonds $e_0,\dots,e_{m-1}$, it is easy to see 
 that under this one-to-many map the measure of each set is multiplied by a factor which is larger than or equal to
\be  
\Gamma_{x,y}  \ = \ \prod_{j=1}^m  \min \left \{ \frac{ \sinh (\beta J_{e_j}) }{ \cosh (\beta J_{e_j})}, \frac{  \cosh (\beta J_{e_j}) - 1}{  \sinh (\beta J_{e_j})} \right \}  \,. 
\ee  
(i.e. the original measure multiplied by $\Gamma_x  $ is dominated by the  measure of the image set.) 
This allows us to conclude:  
\begin{multline}  \label{eq:202}
 \sum_{\substack{\n_1\in\Omega_{\Lambda_L}:\,\partial \n_1\ =\  \emptyset   \\ \n_2\in\Omega_{\Lambda_L\cup\{\delta\}}:\,\partial \n_2=\{x,y\}}}w_\beta(\n_1)w_\beta(\n_2)\indf{ x\lr{\widehat{\n_1+\n_2}}\delta }  \ \le  \   \\  
 \le \   \Gamma_{x,y}^{-1} 
\sum_{\substack{\n_1\in\Omega_{\Lambda_L}:\,\partial \n_1=\emptyset\\ \n_2'\in\Omega_{\Lambda_L\cup\{\delta\}}:\,\partial \n_2'=\emptyset}}w_\beta(\n_1)w_\beta(\n_2')\,\indf{ x\lr{\widehat{\n_1+\n_2'}}\delta }\, .
\end{multline} 

Inserting this   in \eqref{eq:57}, we find that 
\begin{align}  \label{eq:489} 
 \langle\sigma_x\sigma_y\rangle_{\Lambda_L,\beta}^+-\langle\sigma_x\sigma_y\rangle_{\Lambda_L,\beta}^0 & 
\  \le \    \Gamma_{x,y}^{-1}   \ {\rm P}^0_{\Lambda_L,\beta}\otimes {\rm P}^+_{\Lambda_L,\beta}[x\lr{\widehat{\n_1+\n_2}} \delta]\, .   
\end{align}
Taking the limit $L\to \infty$ (which exists by  Theorem~\ref{def:current infinite}) we  obtain
\begin{equation}
\label{eq:99}
0 \ \le \   \langle\sigma_x\sigma_y\rangle_{\beta}^+-\langle\sigma_x\sigma_y\rangle_{\beta}^0 \  \le\  \  \Gamma_{x,y}^{-1}    \ \mathbb P_\beta\left[x\leftrightarrow\infty\right]\,. 
\end{equation}
(The percolation event on the right does not depend on finitely many edges, but justifying passing to the limit is straightforward by first considering the events that  $x$ is connected to distance $N$.) 

\medbreak

We now consider $\beta$ for which \eqref{eq:cont_cond}  holds.  Applying Theorem~\ref{thm:percolation} we conclude that $ \mathbb P_\beta\left[x\leftrightarrow\infty\right] =0$, and hence 
for any $x,y\in \bbZ^d$:  $\langle\sigma_x\sigma_y\rangle_{\beta_c}^+ = \langle\sigma_x\sigma_y\rangle_{\beta_c}^0$.  Thus, using the FKG inequality \cite{FKG71} and \eqref{eq:99}: 
\be 
0\le \langle\sigma_0\rangle_{\beta_c}^+\langle\sigma_x\rangle_{\beta_c}^+\le\langle\sigma_0\sigma_y\rangle_{\beta_c}^+ = \langle\sigma_0\sigma_x\rangle_{\beta_c}^0 
\ee 
for any $y\in\bbZ^d$.
The assumption that $\langle\sigma_0\sigma_x\rangle_{\beta_c}^0$ averages to zero over translations  leads to $\langle\sigma_0\rangle_{\beta_c}^+=0$, i.e. $m^*(\beta_c)=0$.   

 The full statement of continuity of the Gibbs state then follows by the known general result which is presented as Proposition~\ref{prop:GScont} in Appendix~\ref{app:continuity}, and the observation that for translation invariant  models $m^*(\beta_c)=0$ implies that $\langle \sigma_x\rangle_{\beta_c}^+ =0$ for all sites $x\in \bbZ^d$.   \\\ 
\qed

\subsection{Proof of Corollary~\ref{cor:main}}\label{sec:3}  

We include the proof of Corollary~\ref{cor:main} only for completeness, as the argument is not new.  
  
The Gaussian domination bound \eqref{eq:89}  (also known as the {\it infrared} bound) can be equivalently formulated as the statement  that for any   function  
 $(v_x)\in\bbC^{\bbT_L}$, the correlation function $F_{L,\beta}(x,y) =  \langle \sigma_x\sigma_y\rangle_{\bbT_L,\beta} $  satisfies:
\be  \label{eq:finitebound}
\sum_{x,y\in \bbT_L}v_x\overline{v_y}F_{L,\beta}(x,y)\le \ \frac{1}{2\beta}   \sum_{x,y\in \bbT_L}v_x\overline{v_y}G_L(x,y) 
\ + \  \frac{1}{L^d}    \widehat F_{L,\beta}(0)   \left| \sum_{x\in  \bbT_L} v_x \right|^2  
~, 
\ee 
where 
\be  \label{eq:GL} 
G_L(x,y)=\sum_{p \  \in \ \bbT_L^\star\setminus\{0\}} 
\frac{1}{L^d}  \   \  \  \frac{e^{i \,  p \cdot (x-y) }}{E(p)} 
~ .
\ee 
 
The  sum in \eqref{eq:GL} (with the weights $\frac{1}{L^d} $) forms a Riemann approximation. 
Under the assumed Condition \eqref{eq:Ep_condition} ($L^1$ integrability of $1/E(p)$) standard approximation arguments allow to conclude  the pointwise convergence 
\be 
\lim_{L\rightarrow \infty}G_L(x,y) \  =\  \int_{[-\pi,\pi]^d}\frac{{\rm d}p}{(2\pi)^d}\frac{e^{i \, p \cdot (x-y) }} {E(p)} \ =:\ G(x,y)
~ , 
\ee 
to a function satisfying:  
\be \label{eq:63} 
\lim_{N\rightarrow \infty}\frac {1}{|\Lambda_N|^2}\sum_{x,y\in \Lambda_N}G(x,y)\  =\  0 \, .
\ee 

Furthermore,    for  $\beta < \beta_c$ and any fixed function $v$ of bounded support the  zero momentum  term in \eqref{eq:finitebound} can be omitted  since
 \begin{eqnarray}  \label{eq:67}
0\  \le \   \lim_{L\to \infty}  \frac{1}{L^d}    \widehat F_{L,\beta}(0)  & = & 
  \lim_{L\to \infty}    \frac{1}{|\bbT_L|} \sum_{x\in \bbT_L} \langle\sigma_0\sigma_x\rangle_{\bbT_L,\beta}  \notag \\ 
 & \le &   \lim_{L\to \infty}    \frac{1}{|\Lambda_L|} \sum_{x\in \Lambda_L} \langle\sigma_0\sigma_x\rangle_{\Lambda_L,\beta}^+ \ = \ 0    \, ,
\end{eqnarray}  
where use is made of the fact that  $\sum_{x\in \bbZ^d } \langle\sigma_0\sigma_x\rangle_{\beta}^+ < \infty$ for any $\beta < \beta_c$ (\cite{ABF87}).

To apply the above to the free boundary condition correlation function $\langle \sigma_x\sigma_y\rangle_{\Lambda_L,\beta}^0 $ let us first note that, by the  Griffith inequality \cite{Gri67},  the latter are monotone increasing functions of $\beta$ and also monotone increasing in $L$.    Standard semicontinuity arguments which are applicable under such monotonicity assumptions allow to conclude that  for each $x,y\in \bbZ^d$   
\be  \label{eq:77}
   \langle\sigma_x\sigma_y\rangle_{\beta_c}^0 \ = \  
\lim_{\beta  \nearrow \beta_c} \   \lim_{L\to \infty}   \langle\sigma_x\sigma_y\rangle_{\Lambda_L, \beta}^0 \   
\ee

Combining  \eqref{eq:77} with  \eqref{eq:finitebound} for the function $v_x = \frac{1}{|\Lambda_n|}  {\bf 1}[ x\in\Lambda_n] $, and using  \eqref{eq:67},  we find that for each fixed $n$: 
\begin{eqnarray}  \label{eq:65}
\frac{1}{|\Lambda_n|^2}  \sum_{x,y\in\Lambda_n}\langle\sigma_x\sigma_y\rangle_{\beta_c}^0  & =  &  
\lim_{\beta  \nearrow \beta_c} \   \lim_{L\to \infty}  \frac{1}{|\Lambda_n|^2}  \sum_{x,y\in\Lambda_n}\langle\sigma_x\sigma_y\rangle_{\Lambda_L, \beta}^0 \   \notag 
    \\[2ex]  
 & \le  &  
\lim_{\beta  \nearrow \beta_c} \   \lim_{L\to \infty}  \frac{1}{|\Lambda_n|^2}  \sum_{x,y\in\Lambda_n}\langle\sigma_x\sigma_y\rangle_{\bbT_L, \beta} \notag
 \\[3ex]  
 & \le   &     \frac1{2\beta_c} \   \frac{1}{|\Lambda_n|^2}  \sum_{x,y\in \Lambda_n}G(x,y)   \   ~.
\end{eqnarray}   
where another use is made of the  Griffith inequality \cite{Gri67}, by which  $\langle\sigma_x\sigma_y\rangle_{\Lambda_L, \beta}^0\le \langle\sigma_x\sigma_y\rangle_{\bbT_L, \beta}$.  
Incorporating now \eqref{eq:63} in \eqref{eq:65}, we see that if Condition \eqref{eq:Ep_condition}  holds,  then $\widehat M_{LRO}(\beta_c) = 0$.  Thus, by Theorem~\ref{thm:continuity} the spontaneous magnetization $m^*(\beta)$ vanishes  continuously at $\beta_c$,  as claimed in Corollary~\ref{cor:main}. 
\qed

\vspace{1cm} 

\appendix 
\noindent {\LARGE \bf Appendix}

\section{The mixing properties of the two Gibbs states} \label{app:ergodicity}

\begin{proposition}  \label{prop:mixing} 
For any translation invariant ferromagnetic Ising model on $\bbZ^d$ satisfying conditions $\bf{ C1 - C4}$:
\begin{enumerate} 
\item 
the state  $\big \langle  \cdots   \big \rangle^{+}_\beta$  is ergodic and mixing, in the sense that for any pair of local functions  $F,G: \{-1,1\}^{\bbZ^d} \mapsto \bbR $ the following limit exists and satisfies:
\be  \label{eq:mixing} 
\lim_{\|x\| \to \infty}  \big \langle F\times G \circ \tau_x   \big \rangle^{+}_\beta \ = \ 
\big \langle  F  \big \rangle^{+}_\beta\  \times \    
\big \langle  G  \big \rangle^{+}_\beta \, . 
\ee
\item The state  $\big \langle  \cdots   \big \rangle^0_\beta$  is ergodic and mixing in its restriction to the $\sigma$-algebra of even events, i.e. it satisfies the analog of   \eqref{eq:mixing}   for
functions such that $F(-\sigma) = F(\sigma)$ and $G(-\sigma) = G(\sigma)$.
\end{enumerate} 
\end{proposition} 

The proof is based on the Griffith inequality \cite{Gri67} which implies the   
monotonicity, in the coupling strength, of the expectation values of  $\sigma_A := \prod_{x\in A} \sigma_x$ in ferromagnetic Ising spin systems with Hamiltonians of the form 
 \be 
 H(\sigma) \ = \  -  \sum_{B\subset \Omega_2} J_B  \, \sigma_B \, , 
 \ee 
with $J_B \ge 0$ for all finite subset $B$ of $\Omega_2$, and $\sigma_B := \prod_{x\in B} \sigma_x $.  Note that the Hamiltonians satisfying $\bf{ C1 - C4}$ are of this form with $J_B=0$ for any set of cardinality different from 2.

\begin{proof} Let $F=\sigma_A$ with $A$ finite, and let $G=e^{-\widetilde H}$, with  $\widetilde H $ given by a finite sum of the form  $\widetilde H=-\sum_{B\subset\bbZ^d}    \widetilde J_B \sigma_B$. We have
$$
 \big \langle \sigma_A \,  e^{-\beta\widetilde H \circ \tau_x } \big \rangle^\#_{ \beta }  \ = \   
 \big \langle \sigma_A \,  \big | e^{-\beta\widetilde H \circ \tau_x  } \big \rangle^{\#}_{\beta }  \times 
 \big \langle  e^{-\beta\widetilde H \circ \tau_x }  \, \big \rangle^\#_{ \beta } \ = \   
 \big \langle \sigma_A \,  \big | e^{-\beta\widetilde H \circ \tau_x  } \big \rangle^{\#}_{\beta }  \times 
 \big \langle  e^{-\beta\widetilde H}  \, \big \rangle^\#_{ \beta } \, ,
$$
where the first term on the right-hand side can be interpreted as an Ising measure with coupling constants equal to $J_B+\widetilde J_{B-x}$. Now, the above mentioned monotonicity,  which follows from the 
 Griffith inequality \cite{Gri67}, implies that
\begin{itemize}[nolistsep]
\item if $\widetilde J_B\ge0$ for any $B\subset\bbZ^d$, then
\be  \label{eq:bracket_+}
\big \langle \sigma_A \big \rangle^{+}_{\Lambda_{ \|x\|_1 /2 ,\beta }  } \ \ge \ 
\big \langle \sigma_A \, \big  |  e^{ -\beta\widetilde H \circ \tau_x  }  
\big \rangle^{+}_\beta    \ \ge \ 
\big \langle \sigma_A \big \rangle^{+}_\beta  \, .
\ee 
\item if $-J_B\le \widetilde J_B\le 0$ for any finite subset $B$ of $\bbZ^d$, then
\be  \label{eq:bracket_f}
\big \langle \sigma_A \big \rangle^0_{\Lambda_{ \|x\|_1 /2 ,\beta }  } \ \le \ 
\big \langle \sigma_A \, \big  |  e^{ -\beta\widetilde H \circ \tau_x  }  
\big \rangle^0_\beta    \ \le \ 
\big \langle \sigma_A \big \rangle^0_\beta  \, .
\ee 
\end{itemize}
The convergence of $\langle  \cdots \rangle^{\#}_{\Lambda_{ \|x\|_1 /2 ,\beta } }$ to $\langle  \cdots \rangle^{\#}_\beta$ thus implies that
$$
\lim_{\|x\| \to \infty} 
\big \langle  \sigma_A |   e^{-\beta\widetilde H  \circ \tau_x  }  \big \rangle^{\#}_{\beta }  \  = \   
 \big \langle \sigma_A \big \rangle^{\#}_{\beta }$$
and therefore  
\be\label{eq:200}\lim_{\|x\| \to \infty} 
\big \langle \sigma_A \,  e^{-\beta\widetilde H \circ \tau_x } \big \rangle^\#_{ \beta }  \  = \   
 \big \langle \sigma_A \big \rangle^{\#}_{\beta }\big \langle e^{-\widetilde H} \big \rangle^\#_{ \beta }\ee
for $\widetilde J_B\ge0$ ($\forall B\subset\bbZ^d$) in the case of $+$ boundary conditions, and $-J_B\le \widetilde J_B\le 0$ ($\forall B\subset\bbZ^d$) for free boundary conditions. 

By linearity, in equation \eqref{eq:200} the term $\sigma_A$ may be replaced by arbitrary polynomials in $\{\sigma_x\}$ and thus, by the Stone-Weierstrass theorem, the statement extends to all continuous functions $F(\sigma)$, and through that to all bounded measurable functions.  Similarly, using also the convergence of the  power series expansion,   
the factor $e^{-\beta\widetilde H}$ can be replaced by  bounded measurable functions which are spanned by the collection of terms which are allowed in $\widetilde H$.  For $+$ boundary conditions the latter includes $\sigma_A$ for all bounded sets, and thus  \eqref{eq:200}  extends to mixing, and hence  ergodicity, of the $+$ state.   For the free boundary conditions the restriction in \eqref{eq:bracket_f}, 
 which does not allow $\widetilde J_B$  to overturn the ferromagnetic nature of the state,  excludes odd functions.  However the argument still allows to conclude the mixing property on the  $\sigma$-algebra generated by the collection of random variables $\{ \sigma_{\{x,y\}} \, : \,  J_{x,y} >0 \}$.  That is easily seen to consist of  the $\sigma$-algebra of even events. 
\end{proof}

 \medskip 

\section{From the continuity of the magnetization to the continuity  of the Gibbs state(s)} 
 \label{app:continuity} 

In the last step of the proof of Theorem~\ref{thm:continuity} use was made of the following known property of the Ising model's Gibbs states~\cite{LML72}.  It is presented here for completeness, and in a somewhat streamlined form  which experts would also find familiar.
\begin{proposition} \label{prop:GScont}
In any ferromagnetic Ising model with pair interactions, on a graph $\bbG$  with  $\sum_{y\in \bbG} J_{x,y}  < \infty$ for all $x\in \bbG$, for $h=0$ and any  $\beta_0 \in [0,\infty)$ the following conditions are equivalent
\begin{enumerate} 
\item[(a)] for all $x\in \bbG$:  $\langle \sigma_x \rangle^+_{\beta_0} = 0$, 
\item[(b)]  there is a  unique   Gibbs state at  $\beta_0$, as well as all $\beta < \beta_0$,  
and the Gibbs state is continuous in the sense that for any choice of boundary conditions, which may also vary with $\beta$ for $\beta>\beta_0$: 
\be 
\lim_{\beta \to \beta_0}  \langle \cdots \rangle^{\#(\beta)}_\beta  \ =\ \langle \cdots \rangle_{\beta_0}     \, . 
\ee  
\end{enumerate}  
\end{proposition} 
For clarity: the statement refers to the Gibbs states in the limit in which the graph has no boundary sites. 
\begin{proof} 
The relations are based on two   properties of the systems' Gibbs state.   First is that at all $\beta$ the Gibbs states are ``bracketed'' in the sense of Fortuin-Kasteleyn-Ginibre~\cite{FKG71} between the $\pm $ boundary condition states and the two 
are equal if and only if  condition (a) holds (\cite{LML72}).   Thus (a) is equivalent to the uniqueness of the Gibbs state.  

The added statement of the Gibbs state's continuity is based on semi-continuity arguments, which imply that for all finite $A\subset \bbG$: : 
\begin{eqnarray} 
\langle \sigma_A \rangle^{+}_\beta  & = &  \lim_{\varepsilon \searrow 0}  \langle \sigma_A \rangle^{+}_{\beta+\varepsilon}\, \,  ,  \notag \\ 
\\
\langle \sigma_A\rangle^{0}_\beta  & = &  \lim_{\varepsilon \searrow 0}  \langle \sigma_A \rangle^{0}_{\beta-\varepsilon}  \, \, . 
\notag
\end{eqnarray} 
The proof of these relations is based on the Griffith inequalities, by which the finite volume functions 
$\langle \sigma_A \rangle^{\#}_{\Lambda,\beta}$ are: {\em i.} monotone increasing in $\beta$, and {\em ii.}  monotone in $\Lambda\subset \bbG$: increasing for $\#=0$ and decreasing for $\#=+$.   This makes applicable  the general semicontinuity principle that pointwise monotone limits of sequences of continuous functions are upper semicontinuous in the increasing case, and lower semicontinuous in the decreasing case.    
 
Thus, under the assumption  of the uniqueness of state for $\beta_0$, which directly extends to uniqueness for all $\beta \le  \beta_0$, we have: 
\be  
\langle \cdots \rangle_{\beta_0}   \  = \  \lim_{\varepsilon \to  0}  \langle \cdot \rangle^{+}_{\beta+\varepsilon}  
\ = \   \lim_{\varepsilon \to  0}  \langle \cdot \rangle^{-}_{\beta+\varepsilon} 
\ee  
which implies  the full statement of continuity throughout the FKG bracketing principle.  \\ 
\end{proof}

\noindent {\bf Remark:}  Related to the above is the afore mentioned relation of~\cite{Leb77}  between  the continuity of the states $\langle \cdot \rangle^0_\beta$ and the differentiability of the free energy in $\beta$ (for graphs satisfying the van Hove growth condition).   
It would be of interest to see an unconditional derivation of    \eqref{eq:f=pm} 
which may well be valid for  translation invariant models at  all $\beta$.

\paragraph{Acknowledgements}
We thank Vincent Tassion, Gady Kozma, Nick Crawford and Marek Biskup for stimulating discussions.  We also thank Christophe Garban and Yvan Velenik for  carefully reading the manuscript. This project was completed during a stay in the Weizmann Institute and we thank the institution for its kind hospitality and for providing a perfect environment for research.  
The work was support in parts by Simons Fellowship and NSF grant PHY-1104596  (MA),  ERC grant AG CONFRA, FNS and NCCR SwissMap 
(HD),
and the Weizmann Institute's Visiting Scientists Program.
 
\addcontentsline{toc}{chapter}{Bibliography}   


\providecommand{\bysame}{\leavevmode\hbox to3em{\hrulefill}\thinspace}
\providecommand{\MR}{\relax\ifhmode\unskip\space\fi MR }
\providecommand{\MRhref}[2]{%
  \href{http://www.ams.org/mathscinet-getitem?mr=#1}{#2}
}
\providecommand{\href}[2]{#2}

\begin{flushright}
\footnotesize\obeylines
     \textsc{Departments of Physics and Mathematics \\ 
              Princeton University\\  Princeton, United States \\ 
              E-mail: \texttt{aizenman@princeton.edu} }
   \medbreak
\textsc{D\'epartement de Math\'ematiques \\ 
             Universit\'e de Gen\`eve \\  Gen\`eve, Switzerland \\ 
       {E-mail:} \texttt{hugo.duminil@unige.ch}  }
  \medbreak
   \textsc{IMPA \\ 
   Rio de Janeiro, Brazil \\ 
  \textsc{E-mail:} \texttt{vladas@impa.br} }
\medbreak
\end{flushright}

\end{document}